\documentclass[11pt,letterpaper]{article}

\usepackage{amsfonts}
\usepackage{amsmath}
\usepackage{amssymb}
\usepackage{enumerate}
\usepackage{natbib}
\usepackage{color}
\usepackage{graphicx,subfigure}
\usepackage{setspace}
\usepackage{epstopdf}
\usepackage{mathtools}                
\mathtoolsset{showonlyrefs=true}
\usepackage{fixltx2e,amsmath}                                           

\MakeRobust{\eqref}
\usepackage{amssymb}                            
\usepackage[mathscr]{eucal}                     
\usepackage{setspace}
\usepackage{hyperref}
\newcommand\E{\ensuremath{\mathbb{E}}}
\newcommand\R{\ensuremath{\mathbb{R}}}

\newcommand\F{\mathcal{F}}

\newcommand\half{\frac{1}{2}}

\newcommand{\indic}[1]{1\hspace{-2.1mm}{1}_{\{#1\}}} 

\def\R{{\mathbb R}}

\def\P{{\mathbb P}}

\def\1{{\mathbf 1}}
\def\F{{\mathcal F}}

\def\L{{\mathcal L}}

\def\setT{{\mathcal T}}
\newtheorem{theorem}{Theorem}

\newtheorem{lemma}[theorem]{Lemma}

\newtheorem{remark}[theorem]{Remark}

\newenvironment{proof}[1][Proof]{\noindent\textbf{#1.} }{\ \rule{0.5em}{0.5em}}
\numberwithin{equation}{section}
\numberwithin{theorem}{section}
\begin{document}
\title{Optimal Starting-Stopping and Switching of a CIR Process\\ with Fixed Costs}
\author{Tim Leung\thanks{IEOR Department, Columbia University, New York, NY 10027; \,\mbox{tl2497@columbia.edu}. Corresponding author. } \and Xin Li\thanks{IEOR Department, Columbia University, New York, NY 10027; \,\mbox{xl2206@columbia.edu}.} \and Zheng Wang\thanks{IEOR Department, Columbia University, New York, NY 10027; \,\mbox{zw2192@columbia.edu}.}}
\date{\today}
\maketitle
\begin{abstract}
This paper analyzes the problem of starting and stopping a Cox-Ingersoll-Ross (CIR) process with fixed   costs. In addition, we  also study a related optimal switching problem that involves  an infinite sequence of starts and stops. We establish the conditions under which the starting-stopping and switching problems admit the same optimal starting and/or stopping strategies.  We rigorously prove that the optimal starting and stopping strategies are of threshold type, and give the analytical expressions for the   value functions in terms of confluent hypergeometric functions. Numerical examples are provided to illustrate the dependence of timing strategies on  model parameters and transaction costs. 
\end{abstract}\vspace{10pt}

\begin{small}
\end{small}

\begin{small}
\tableofcontents
\end{small}

\onehalfspacing

\section{Introduction}
Many problems in  risk analysis and  finance involve a sequence of starting and stopping decisions. For instance,  a trader may seek the optimal times  to  buy and then  sell an asset, or a firm may seek to invest in a project and subsequently sell it to another firm.  This motivates  us to study an optimal starting-stopping (or entry-exit) timing problem.

In general, the optimal timing    depends crucially  on the dynamics of the underlying process. In particular,  if  the process is  a super/sub-martingale, such as the geometric Brownian motion, then it is optimal either to  stop  immediately or wait forever (see e.g. \cite{shiryaev2008thou}). As such, the corresponding optimal single stopping, starting-stopping, or switching problems will admit a trivial solution.  In this paper, our optimal starting-stopping and switching problems are driven by the Cox-Ingersoll-Ross (CIR) process (see \cite{CIR85}). The CIR process has been widely used as the model for interest rates,  volatility, and energy prices (see, for example,  \cite{CIR85},  \cite{heston93}, \cite{CIRCommodity}, respectively).   In the real option literature, mean reverting processes have  been used to model the value of a project. For instance, \cite{ewald2010irreversible} studies the timing of an   irreversible investment whose value is  modeled by a CIR process, and they  solve the associated   optimal single stopping problem. \cite{carmona2007investment} examine the optimal investment timing  where the  interest rate evolves according to  a CIR process. Other  mean-reverting models, such as the exponential  Ornstein-Uhlenbeck (OU) model, have been discussed in  \cite{dixit1994investment}.  In contrast to these studies, our model addresses multiple timing decisions.

 Our main contribution is the  analytical derivation of    the non-trivial optimal starting and stopping timing strategies and the associated value functions.  Under both starting-stopping and switching approaches, it is optimal to stop when the process value is sufficiently high, though at different levels.  As for starting timing,  we find the necessary and sufficient conditions whereby it is optimal not to start at all when facing the optimal switching problem. In this case,  the optimal switching problem in fact reduces  into an optimal single stopping problem, where  the optimal stopping level is identical to that of  the  optimal starting-stopping problem.  

In the literature,  \cite{menaldi1996optimal} study an optimal starting-stopping problem for general Markov processes, and provide the mathematical characterization of the value functions.   \cite{sun1992nested} considers an optimal starting-stopping problem driven by a time-homogeneous diffusion, and analyzes the associated nested variational inequalities (VIs).  \cite{zervos2011buy} discusses a solution methodology  involving coupled variational inequalities   for an optimal switching problem under some time-homogeneous diffusions subject to  fixed transaction costs. \cite{tanaka1993} studies an optimal starting-stopping problem for two-parameter stochastic processes under discrete time.

The VI approach is commonly used  not only for optimal starting-stopping and switching problems, but also for optimal stopping problems (see \cite{Bensoussan} and \cite{Oksendal2003}, among others).   The  VI method is most useful when the  reward and value  functions are simple explicit functions with amenable properties. In our starting-stopping problem, however, the reward function for the starting problem is expressed in terms of the confluent hypergeometric function of the first kind and can be positive or negative. This motivates us to adopt a probabilistic approach to rigorously   construct the value functions.    In contrast to the VI approach,  we solve  the optimal starting-stopping problem by characterizing the  value functions as the decreasing smallest \textit{concave majorant} of the corresponding reward functions. This allows us to directly derive  the value function and prove its optimality, without \emph{a priori} conjecturing the resulting regions for starting/stopping/continuation. This methodology has been discussed  in  \cite{dayanik2003optimal}, dating back to \cite{dynkin1969theorems}.     Having solved the optimal starting-stopping problem, we  then apply our results   to infer a similar solution structure for the optimal switching problem and verify using the variational inequalities. 

As for related applications of optimal starting-stopping problems,   \cite{LeungLi2014OU} study the optimal timing to trade a mean-reverting price spread with a stop-loss constraint. \cite{LeungLiWang2014XOU} investigate the problem of buying and selling  an exponential-OU underlying subject to transaction costs. \cite{KS2007} analyze  the starting-stopping times for a risk process for an insurance firm.  In the context of   derivatives trading,  \cite{LeungLiu} study the  optimal timing to liquidate credit derivatives  where the default intensity is modeled by an OU or CIR  process. They focus on the finite-horizon trading problem, and   identify the conditions under which immediate stopping or perpetual holding is optimal. 

The rest of the paper is structured as follows. In Section \ref{sect-overview}, we formulate both the optimal starting-stopping and optimal switching problems. Then, we present our analytical results and numerical examples  in Section \ref{sect-solution}. The   proofs of our main  results are detailed in Section \ref{sect-method}. Finally, the Appendix contains the proofs for   a number of lemmas.

\section{Problem Overview}\label{sect-overview}
Denote the probability space by $(\Omega, \F, \P)$, where $\P$ is   the historical probability measure. We consider a CIR process $(Y_t)_{t \geq 0}$ that satisfies the SDE
\begin{align}\label{YCIR}
dY_{t}=  \mu( \theta - Y_t)\,dt+\sigma \sqrt{Y_t} \,dB_{t},
\end{align}
with constants $\mu, \theta, \sigma>0,$ and $B$ is a standard Brownian motion defined on the probability space. If  $2\mu\theta \ge   \sigma^2$ holds, which is often referred to as the Feller condition (see \cite{feller1951two}), then the level $0$ is inaccessible by $Y$. If  the initial value $Y_0>0$, then $Y$ stays strictly positive at all times almost surely. Nevertheless,  if $Y_0 =0$, then $Y$ will enter  the interior of the state space immediately and stays positive thereafter almost surely.   If $2\mu\theta <  \sigma^2$, then the level $0$  is a reflecting boundary. This means that once $Y$ reaches $0$, it  immediately  returns to the interior of the state space and continues to evolve. For a detailed categorization of boundaries for diffusion processes, we refer to  Chapter 2 of \cite{borodin2002handbook} and Chapter 15 of \cite{karlin1981second}.


\subsection{Optimal Starting-Stopping Problem}
Given a CIR process, we first consider the optimal timing to stop.  If a decision to stop is made at some time $\tau$, then the amount  $Y_\tau$ is received and simultaneously the constant transaction cost $c_s > 0$ has to be paid. Denote by $\mathbb{F}$ the filtration generated by $B$, and $\setT$    the set of all $\mathbb{F}$-stopping times.  The maximum expected discounted value is obtained by solving the optimal stopping problem
\begin{align}\label{VCIR} 
V(y) &= \sup_{\tau \in \setT}\E_y\!\left\{e^{-r \tau}(Y_{\tau} - c_s)\right\}, 
\end{align}
where   $r>0$ is the constant discount rate, and  $\E_y\{\cdot\}\equiv\E\{\cdot|Y_0=y\}$.

The value function $V$ represents the expected   
value from optimally stopping the process $Y$. On the other hand, the process value plus the
transaction cost constitute  the total cost to start. Before even starting, one needs to choose
the optimal timing to start, or not to start at
all. This leads us to analyze the  starting timing inherent in the starting-stopping 
problem. Precisely, we solve \begin{align}\label{JCIR}J(y) &=  \sup_{\nu \in \setT}\E_y\!\left\{e^{-r \nu} ( V(Y_{\nu})  - Y_{\nu} - c_b)\right\}, \end{align} 
with the constant transaction cost $ c_b> 0$ incurred  at the start.  In other words, the objective is
to maximize the expected difference between the value function
$V(Y_{\nu})$ and the current  $Y_{\nu}$, minus transaction cost  $c_b$. The value function $J(y)$
represents the maximum expected value that can be gained by entering and subsequently exiting, with transaction costs $c_b$ and $c_s$ incurred, respectively, on entry and exit. For  our analysis, the transaction costs $c_b$ and $c_s$ can be different. To facilitate presentation, we denote the functions
\begin{align}\label{hCIR}
h_s(y) = y-c_s, \quad \text{ and } \quad  h_b(y)=y+c_b.
\end{align}

If it turns out that $J(Y_0)\le  0$  for some initial value
$Y_0$, then it is optimal not to start at all. Therefore,  it is important to identify the trivial cases. Under the CIR model, since   $\sup_{y\in\R_+}( V(y) -h_b(y))\leq 0$
implies that   $J(y) = 0$ for $y \in \R_+$, we shall therefore focus on the  case with
\begin{align}\label{generalassume} \sup_{y\in\R_+}( V(y) -h_b(y))> 0,
\end{align} and solve for the non-trivial optimal  timing strategy.

\subsection{Optimal Switching Problem}\label{ISF}
Under the optimal switching approach,  it is assumed that an infinite number of entry and exit actions take place. The sequential entry and exit times are modeled  by the stopping times  $\nu_1,\tau_1,\nu_2,\tau_2,\dots \in \setT$ such that
\begin{align*}
0\leq \nu_1 \leq \tau_1 \leq \nu_2 \leq \tau_2 \leq \dots.
\end{align*}
Entry and exit decisions are made, respectively, at times  $\nu_i$ and   $\tau_i$, $i\in\mathbb{N}$. The  optimal timing to enter or exit would depend  on the  initial position. Precisely, under the CIR model,  if the initial position is zero, then the first task is to determine when to \emph{start} and  the corresponding optimal switching problem is
\begin{align}\label{JCIRsw}
{\tilde{J}}(y) &= \sup_{\Lambda_0}\E_y\left\{\sum_{n=1}^\infty [e^{-r\tau_n}h_s(Y_{\tau_n}) - e^{-r \nu_n} h_b(Y_{\nu_n})]\right\},
\end{align}
with   the set of admissible stopping times $\Lambda_0=(\nu_1,\tau_1,\nu_2,\tau_2,\dots)$, and  the reward functions $h_b$ and $h_s$ defined in \eqref{hCIR}.  On the other hand, if we start with a long position, then it is necessary to solve
\begin{align}\label{VCIRsw}
{\tilde{V}}(y) &= \sup_{\Lambda_1}\E_y\left\{e^{-r\tau_1}h_s(Y_{\tau_1}) + \sum_{n=2}^\infty [e^{-r\tau_n}h_s(Y_{\tau_n}) - e^{-r \nu_n} h_b(Y_{\nu_n})] \right\},
\end{align}
with  $\Lambda_1=(\tau_1,\nu_2,\tau_2,\nu_3,\dots)$ to determine when to \emph{stop}.


 In summary, the optimal  starting-stopping and switching problems differ in the number of entry and exit decisions.  Observe that any strategy for the starting-stopping problem \eqref{VCIR}-\eqref{JCIR} is also a candidate strategy for the switching problem \eqref{JCIRsw}-\eqref{VCIRsw}. Therefore, it follows that  $V(y) \leq {\tilde{V}}(y)$ and  $J(y) \leq {\tilde{J}}(y).$ Our objective is to derive and compare the corresponding optimal timing strategies under these two approaches.  

\section{Summary of Analytical Results}\label{sect-solution}
We first summarize our analytical results and illustrate the optimal starting and stopping strategies. The  method of solutions and their proofs will be discussed in Section \ref{sect-method}.
We consider the optimal starting-stopping problem followed by the optimal switching problem. First, we denote the infinitesimal generator of $Y$ as
\begin{align}\label{LCIR}
\L &= \frac{\sigma^2y}{2}\frac{d^2}{d y^2} + \mu(\theta - y)\frac{d}{d y},
\end{align}
and consider the ordinary differential equation (ODE)
\begin{align}\label{ODE}
\L u(y)=ru(y), \quad \textrm{ for } \, y \in \R_+.
\end{align}
To present the solutions of this ODE, we define the functions
\begin{align}\label{FGCIR}
F(y) := M(\frac{r}{\mu},\frac{2\mu\theta}{\sigma^2}; \frac{2\mu y}{\sigma^2}),\quad \text{and} \quad
G(y) := U(\frac{r}{\mu},\frac{2\mu\theta}{\sigma^2}; \frac{2\mu y}{\sigma^2}),
\end{align}
where
\begin{align*}
M(a,b;z) &= \sum_{n=0}^{\infty} \frac{a_n z^n}{b_n n!}, \qquad a_0=1,~ a_n= a(a+1)(a+2)\cdots(a+n-1),\\
U(a,b;z) &= \frac{\Gamma(1-b)}{\Gamma(a-b+1)}M(a,b;z)+\frac{\Gamma(b-1)}{\Gamma(a)}z^{1-b}M(a-b+1,2-b;z)
\end{align*}
are the confluent hypergeometric functions of first and second kind, also called the Kummer's function and Tricomi's function, respectively (see Chapter 13 of \cite{abramowitz1965handbook} and Chapter 9 of \cite{lebedev1972special}). As is well known (see   \cite{going2003survey}),   $F$ and $G$ are strictly positive and, respectively, the strictly increasing and decreasing continuously differentiable solutions of the ODE \eqref{ODE}. Also, we remark that the discounted processes $(e^{-rt}F(Y_t))_{t\ge 0}$ and $(e^{-rt}G(Y_t))_{t\ge 0}$ are martingales. 

In addition, recall the reward functions defined in \eqref{hCIR} and note that 
\begin{align}\label{CIRLrbs}
(\L - r)h_b(y)\begin{cases}
>0 &\, \textrm{ if }\, y<y_b,\\
<0 &\, \textrm{ if }\, y>y_b,
\end{cases}
\quad \textrm{and}\quad
(\L - r)h_s(y)\begin{cases}
>0 &\, \textrm{ if }\, y<y_s,\\
<0 &\, \textrm{ if }\, y>y_s,
\end{cases}
\end{align}
where the critical constants $y_b$ and $y_s$ are defined by
\begin{align}\label{ybys}
y_b := \frac{\mu\theta-rc_b}{\mu+r} \quad \textrm{and} \quad y_s := \frac{\mu\theta+rc_s}{\mu+r}.
\end{align}
Note that $y_b$ and $y_s$ depend on the parameters $\mu,$ $\theta$ and $r$, as well as $c_b$ and $c_s$ respectively, but not $\sigma.$

\subsection{Optimal Starting-Stopping Problem}\label{sect-CIR-exit}
We now present the results for the optimal starting-stopping problem \eqref{VCIR}-\eqref{JCIR}. As it turns out, the value function $V$ is expressed in terms of $F$, and $J$ in terms of $V$ and $G$. The functions $F$ and $G$ also play a role in determining the optimal starting and stopping thresholds.

%

\begin{theorem}\label{thm:optLiquCIR}
The value function for the optimal
stopping  problem \eqref{VCIR} is given by
\begin{align}V(y) =
\begin{cases}
\frac{b^*-c_s}{F(b^*)}F(y) &\, \textrm{ if }\, y\in [0,b^*),\\
y-c_s &\, \textrm{ if }\, y\in [b^*, +\infty).
\end{cases}\label{VCIRsol}
\end{align}
Here,  the optimal stopping level $b^*\in (c_s\vee y_s, \infty)$ is found from the equation
\begin{align}\label{eq:solvebCIR}
F(b) = (b-c_s){F}^\prime\!(b).
\end{align}
\end{theorem}


Therefore, it is optimal to stop as soon as the process $Y$ reaches $b^*$ from below. The stopping level $b^*$ must also be higher than the fixed cost  $c_s$ as well as the critical level $y_s$ defined in \eqref{ybys}.

Now we turn to the optimal starting problem. Define the reward function 
\begin{align}\label{hhat}
\hat{h}(y) :=V(y) - (y + c_b).
\end{align}
Since $F$, and thus $V$, are convex, so is $\hat{h}$, we also observe  that the reward function $\hat{h}(y)$  is decreasing in $y$. To exclude the scenario where it is optimal never to start, the condition stated in \eqref{generalassume}, namely,  $\sup_{y\in \R+} \hat{h}(y)>0$,  is now equivalent to  
\begin{align}\label{condJ}
V(0) = \frac{b^*-c_s}{F(b^*)}>c_b,
\end{align}
since $F(0)=1$.

\begin{theorem}\label{thm:optEntryCIR}
The optimal starting problem \eqref{JCIR} admits the solution
\begin{align}\label{JCIRsol}
  J(y) =
\begin{cases}
  V(y)-(y+c_b) &\, \textrm{ if }\, y \in [0,d^*],\\
  \frac{V(d^*)-(d^*+c_b)}{G(d^*)}G(y) &\, \textrm{ if }\, y \in (d^*, +\infty).
\end{cases}
\end{align}  The optimal starting level $d^* >0$ is uniquely determined from
\begin{align}
\label{eq:solvedCIR}
G(d)({V}^\prime\!(d)  - 1) = G^{'}\!(d)(V(d)-(d+c_b)).
\end{align}
\end{theorem}

As a result, it is optimal to start as soon as the CIR process $Y$ falls below the strictly  positive  level $d^*$.

\subsection{Optimal Switching Problem}\label{sect-switchingCIR}
Now we study the optimal switching problem under the CIR model in \eqref{YCIR}.
We start by giving conditions under which it is optimal not to start ever.
\begin{theorem}\label{thm:CIRth2}
Under the CIR model, if it holds that
\begin{enumerate}[(i)]
\item $y_b\leq 0, \quad \textrm{  or }$
\item $y_b > 0 \quad \textrm{ and } \quad c_b \geq \frac{b^* - c_s} {F(b^*)},$
\end{enumerate}
with $b^*$ given in \eqref{eq:solvebCIR}, then the optimal switching problem \eqref{JCIRsw}-\eqref{VCIRsw} admits the solution
\begin{align}\label{CIRsolsw1_J}
\tilde{J}(y) = 0 \qquad \textrm{for} \quad  y \ge 0,
\end{align}
and
\begin{align}\label{CIRsolsw1_V}
\tilde{V}(y) =
\begin{cases}
\frac{b^* - c_s}{F(b^*)} F(y) &\, \textrm{ if }\, y\in [0,b^*),\\
y-c_s &\, \textrm{ if }\, y\in [b^*, +\infty).
\end{cases}
\end{align}
\end{theorem}

Conditions $(i)$ and $(ii)$ depend on problem data and can be easily verified. In particular, recall that $y_b$ is defined in \eqref{ybys} and is easy to compute, furthermore it is independent of $\sigma$ and $c_s$. Since it is optimal to never enter, the switching problem is equivalent to a stopping problem and the solution in Theorem \ref{thm:CIRth2} agrees with that in Theorem \ref{thm:optLiquCIR}. Next, we provide conditions under which it is optimal to enter as soon as the CIR process reaches some lower level.
\begin{theorem}\label{thm:CIRth3}
Under the CIR model, if 
\begin{align}\label{condJsw}
y_b > 0 \quad \textrm{and} \quad c_b < \frac{b^* - c_s} {F(b^*)},
\end{align}
with $b^*$ given in \eqref{eq:solvebCIR}, then the optimal switching problem \eqref{JCIRsw}-\eqref{VCIRsw} admits the solution
\begin{align}\label{CIRsolsw2_J}
\tilde{J}(y) = \begin{cases}
KF(y)-(y + c_b) &\, \textrm{ if }\, y\in [0,\tilde{d}^*],\\
QG(y) &\, \textrm{ if }\, y\in (\tilde{d}^*,+\infty),
\end{cases} 
\end{align}
and
\begin{align}\label{CIRsolsw2_V}
\tilde{V}(y) =
\begin{cases}
KF(y) &\, \textrm{ if }\, y\in [0,\tilde{b}^*),\\
QG(y)+ (y - c_s) &\, \textrm{ if }\, y\in [\tilde{b}^*, +\infty),
\end{cases}
\end{align}
where
\begin{align}
K=\frac{G(\tilde{d}^*)-(\tilde{d}^*+c_b){G}^\prime\!(\tilde{d}^*)}{{F}^\prime\!(\tilde{d}^*)G(\tilde{d}^*)-F(\tilde{d}^*){G}^\prime\!(\tilde{d}^*)},\\
Q=\frac{F(\tilde{d}^*)-(\tilde{d}^*+c_b){F}^\prime\!(\tilde{d}^*)}{{F}^\prime\!(\tilde{d}^*)G(\tilde{d}^*)-F(\tilde{d}^*){G}^\prime\!(\tilde{d}^*)}.
\end{align}
There exist unique optimal starting and stopping levels $\tilde{d}^*$ and $\tilde{b}^*$, which are found from the nonlinear system of equations:
\begin{align}
\frac{G(d)-(d+c_b){G}^\prime\!(d)}{{F}^\prime\!(d)G(d)-F(d){G}^\prime\!(d)}&= \frac{G(b)-(b-c_s){G}^\prime\!(b)}{{F}^\prime\!(b)G(b)-F(b){G}^\prime\!(b)}, \label{eq:solveCIRG}\\
\frac{F(d)-(d+c_b){F}^\prime\!(d)}{{F}^\prime\!(d)G(d)-F(d){G}^\prime\!(d)}&=\frac{F(b)-(b-c_s){F}^\prime\!(b)}{{F}^\prime\!(b)G(b)-F(b){G}^\prime\!(b)}. \label{eq:solveCIRF}
\end{align}
Moreover, we have that $\tilde{d}^* < y_b \, \textrm{ and } \, \tilde{b}^* > y_s.$
\end{theorem}

In this case, it is optimal to start and stop an infinite number of times where we start as soon as the CIR process drops to $\tilde{d}^*$  and stop when the process reaches $\tilde{b}^*.$ Note that in the case of Theorem \ref{thm:CIRth2} where it is never optimal to start, the optimal stopping level $b^*$ is the same as that of the optimal stopping problem in Theorem \ref{thm:optLiquCIR}. The optimal starting level $\tilde{d^*}$, which only arises when it is optimal to start and stop sequentially, is in general not the same as $d^*$ in Theorem \ref{thm:optEntryCIR}.

We conclude the section with two remarks. 

\begin{remark}
Given the model parameters, in order to identify which of Theorem \ref{thm:CIRth2} or Theorem \ref{thm:CIRth3} applies, we begin by checking whether $y_b\leq 0$. If so, it is optimal not to enter. Otherwise, Theorem \ref{thm:CIRth2} still  applies   if $c_b \geq \frac{b^* - c_s} {F(b^*)}$ holds.  In the other remaining case, the problem is solved as in Theorem \ref{thm:CIRth3}. In fact, the condition $c_b < \frac{b^* - c_s} {F(b^*)}$ implies $y_b>0$ (see the proof of Lemma \ref{lm:hatHCIR} in the Appendix). Therefore, condition \eqref{condJsw} in Theorem \ref{thm:CIRth3} is in fact identical to \eqref{condJ} in Theorem \ref{thm:optEntryCIR}.
\end{remark}

\begin{remark}
To verify the optimality of the results in Theorems \ref{thm:CIRth2} and \ref{thm:CIRth3}, one can show by direct substitution that the solutions  $(\tilde{J}, \tilde{V})$ in \eqref{CIRsolsw1_J}-\eqref{CIRsolsw1_V} and \eqref{CIRsolsw2_J}-\eqref{CIRsolsw2_V} satisfy the variational inequalities:
\begin{align}
\min\{r\tilde{J}(y)-\L \tilde{J}(y), \tilde{J}(y) - (\tilde{V}(y) - (y + c_b))\}&=0,\label{VIJ}\\
\min\{r\tilde{V}(y)-\L \tilde{V}(y), \tilde{V}(y) - (\tilde{J}(y)+ (y - c_s))\}&=0.\label{VIV}
\end{align}
Indeed, this is the approach used by \cite{zervos2011buy} for checking the solutions of their optimal switching problems.
\end{remark}


\subsection{Numerical Examples}
We numerically implement Theorems \ref{thm:optLiquCIR}, \ref{thm:optEntryCIR}, and \ref{thm:CIRth3}, and illustrate the associated starting and stopping thresholds.  In Figure \ref{fig:CIRdb} (left), we observe the changes in optimal starting and stopping levels as speed of mean reversion increases. Both starting levels $d^*$ and $\tilde{d}^*$ rise with $\mu$, from 0.0964 to 0.1219 and from 0.1460 to 0.1696, respectively, as $\mu$ increases from 0.3 to 0.85. The optimal switching stopping level $\tilde{b}^*$ also increases. On the other hand, stopping level $b^*$ for the starting-stopping problem stays relatively constant as $\mu$ changes.

In Figure \ref{fig:CIRdb} (right), we see that as the stopping cost $c_s$ increases, the increase in the optimal stopping levels is accompanied by a fall in optimal starting levels. In particular, the  stopping levels, 
$b^*$ and $\tilde{b}^*$ increase. In comparison, both starting levels $d^*$ and $\tilde{d}^*$ fall. The lower starting level and higher stopping level mean that the entry and exit times are both delayed  as a result of a higher transaction cost. Interestingly, although  the cost $c_s$ applies only when the process is stopped, it also  has an impact on the timing to \textit{start}, as seen in the changes in $d^*$ and $\tilde{d}^*$ in the figure.

In Figure \ref{fig:CIRdb}, we can see that the continuation (waiting) region of the switching problem $(\tilde{d}^*, \tilde{b}^*)$  lies within that of the starting-stopping problem $(d^*, b^*)$. The ability to enter and exit multiple times means it is possible to earn a smaller reward on each individual start-stop sequence while maximizing aggregate return. Moreover, we observe that optimal entry and exit levels of the starting-stopping problem is less sensitive to changes in model parameters than the entry and exit thresholds of the switching problem. 

Figure \ref{fig:CIR_Sim_1} shows a simulated CIR path along with optimal entry and exit levels for both starting-stopping and switching problems. Under the starting-stopping problem, it is optimal to start once the process reaches $d^*= 0.0373$ and to stop when the process hits $b^*=0.4316$. For the switching problem, it is optimal to start once the process values hits $\tilde{d}^*=0.1189$ and to stop when the value of the CIR process rises to $\tilde{b}^*=0.2078$. We note that both stopping levels $b^*$ and $\tilde{b}^*$ are higher than the long-run mean $\theta = 0.2$, and the starting levels $d^*$ and $\tilde{d}^*$ are lower than $\theta$.   The process starts at $Y_0 = 0.15> \tilde{d}^*$, under the optimal switching setting, the first time to enter occurs on day 8 when the process falls to 0.1172 and subsequently exits on day 935 at a level of 0.2105. For the starting-stopping problem, entry takes place much later on day 200 when the process hits 0.0306 and exits on day 2671 at 0.4369. Under the optimal switching problem, two entries and two exits will be completed by the time  a single entry-exit sequence is realized for the starting-stopping problem.

\begin{figure}[ht]
\begin{center}\includegraphics[width=3.1in]{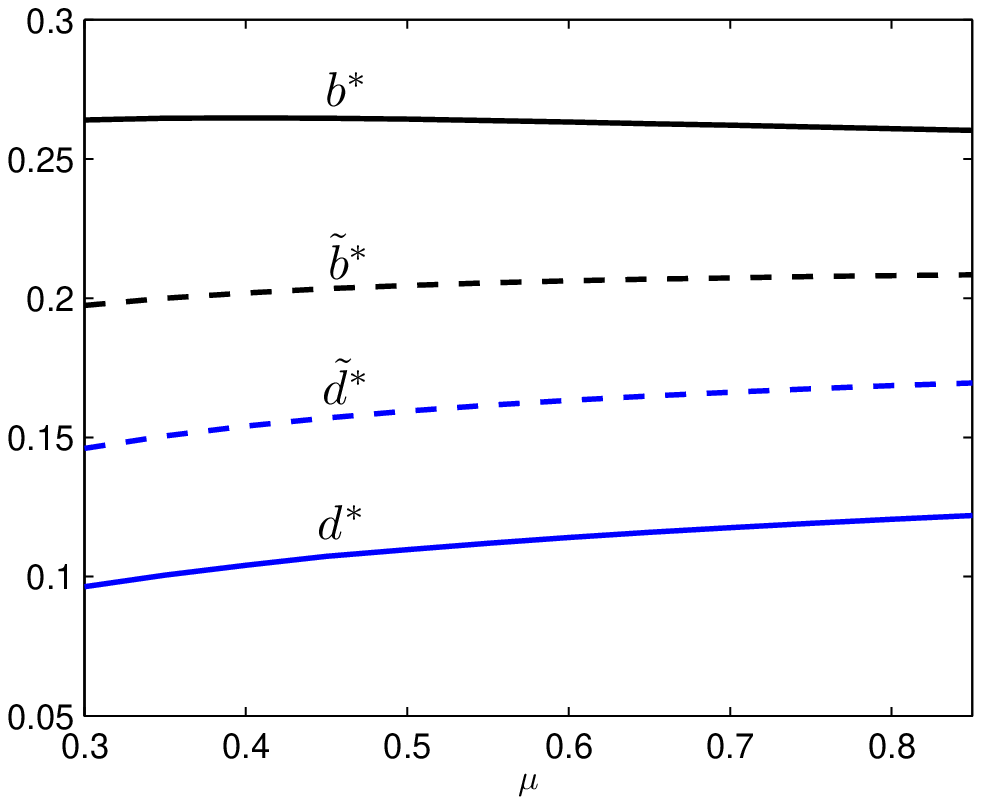}
\includegraphics[width=3.1in]{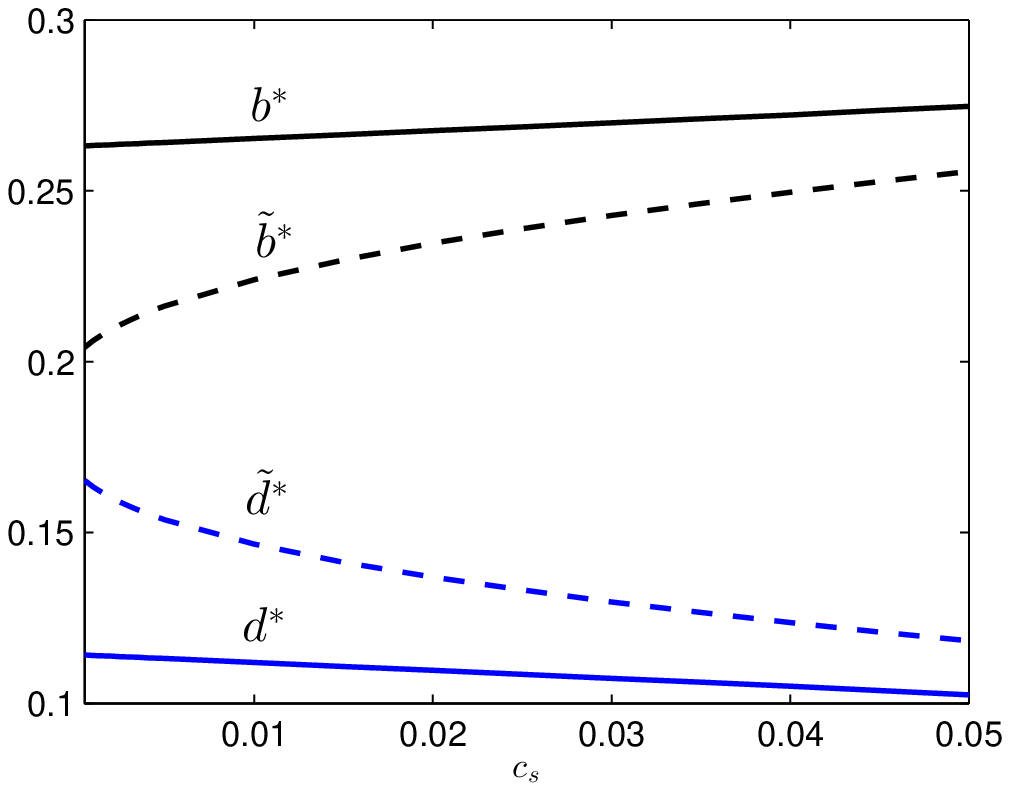}

\caption{\small{(Left) The optimal starting and stopping levels vs speed of mean reversion $\mu$. Parameters: $\sigma = 0.15$, $\theta = 0.2$, $r=0.05$, $c_s=0.001$, $c_b=0.001$.  (Right) The optimal starting and stopping levels vs transaction cost $c_s$. Parameters: $\mu = 0.6$, $\sigma = 0.15$,  $\theta = 0.2$, $r=0.05$, $c_b=0.001$.  }}
\label{fig:CIRdb}
\end{center}\end{figure}

\begin{figure}[t]
\begin{center}
\includegraphics[width=4in]{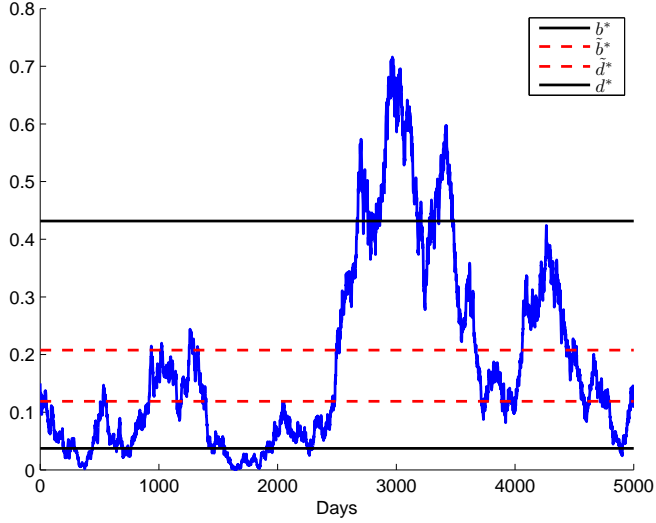}
\caption{\small{A sample CIR path, along with starting and stopping levels. Under the starting-stopping setting, a starting decision is made at $\nu_{d^*} = \inf\{t\geq 0: Y_t \leq d^* = 0.0373\}$, and a stopping decision is made at $\tau_{b^*} = \inf\{ t\geq \nu_{d^*}: Y_t \geq b^*= 0.4316\}$.  Under the optimal switching problem, entry and exit take place at $\nu_{\tilde{d}^*} = \inf\{t\geq 0: Y_t \leq \tilde{d}^* =0.1189\}$ , and $\tau_{\tilde{b}^*} = \inf\{ t\geq \nu_{\tilde{d}^*}: Y_t \geq \tilde{b}^* = 0.2078\}$ respectively. Parameters: $\mu = 0.2$, $\sigma = 0.3$, $\theta = 0.2$, $r=0.05$, $c_s=0.001$, $c_b=0.001$.   }}
\label{fig:CIR_Sim_1}
\end{center}\end{figure}

\section{Methods of Solution and Proofs}\label{sect-method}
We now provide detailed proofs for our analytical results in Section \ref{sect-solution} beginning with the optimal starting-stopping problem. Our main result here is Theorem \ref{thm:VCIR} which provides a mathematical characterization of the value function, and establishes the optimality of our method of constructing the solution.
\subsection{Optimal Starting-Stopping Problem}
We first describe the general solution procedure for the stopping problem $V$,  followed by the starting problem $J.$
\subsubsection{Optimal Stopping Timing}\label{sect-CIRexitproof}
A key step of our solution method  involves   the  transformation
\begin{align}\label{phi}
\phi(y) := -\frac{G(y)}{F(y)}, \quad y \geq 0.
\end{align} With  this, we also  define the function \begin{align}\label{HCIR}
H(z)  := \begin{cases}
\frac{h_s}{F}\circ \phi^{-1}(z) &\, \textrm{ if }\, z<0,\\
\lim_{z\to +\infty}\limits\frac{(h_s(y))^+}{F(y)} &\, \textrm{ if }\, z=0,
\end{cases}
\end{align}
where $h_s$ is given in \eqref{hCIR}. We now prove the analytical form for the value function.
\begin{theorem}\label{thm:VCIR}
Under the CIR model, the value function $V$  of  \eqref{VCIR} is given by
\begin{align}\label{Vchi} V(y)=F(y)W(\phi(y)),\end{align}
with $F$ and $\phi$ given in \eqref{FGCIR} and \eqref{phi} respectively, and $W$ is the decreasing smallest concave majorant of $H$ in \eqref{HCIR}.
\end{theorem}
 \begin{proof}
We first show that $V(y)\geq F(y)W(\phi(y))$. Start at any $y\in [0,+\infty)$, we consider the first stopping time of $Y$ from an interval $[a,b]$ with $0\leq a \leq y \leq b \leq +\infty$. We compute the corresponding expected discounted reward
\begin{align}
\E_y\{e^{-r(\tau_a\wedge\tau_b)}h_s(Y_{\tau_a\wedge\tau_b})\} &= h_s(a)\E_y\{e^{-r\tau_a}\indic{\tau_a<\tau_b}\} + h_s(b) \E_y\{e^{-r\tau_b}\indic{\tau_a>\tau_b}\}\notag\\
&= h_s(a)\frac{F(y)G(b)-F(b)G(y)}{F(a)G(b)-F(b)G(a)} + h_s(b)\frac{F(a)G(y)-F(y)G(a)}{F(a)G(b)-F(b)G(a)}\notag\\
&= F(y)\left[\frac{h_s(a)}{F(a)}\frac{\phi(b)-\phi(y)}{\phi(b)-\phi(a)} + \frac{h_s(b)}{F(b)}\frac{\phi(y)-\phi(a)}{\phi(b)-\phi(a)} \right]\notag\\
&= F(\phi^{-1}(z))\left[H(z_a)\frac{z_b-z}{z_b-z_a}+H(z_b)\frac{z-z_a}{z_b-z_a} \right], \label{EH3CIR}
\end{align}
where  $z_a=\phi(a)$, $z_b=\phi(b)$.
Since $V(y) \geq \sup _{\{a,b: a\le y\le b\}}\E_y\{e^{-r(\tau_a\wedge\tau_b)}h_s(Y_{\tau_a\wedge\tau_b})\}$, we have
\begin{align}\label{VabCIR}
\frac{V(\phi^{-1}(z))}{F(\phi^{-1}(z))} \geq \sup _{\{z_a,z_b:z_a\le z\le z_b\}} \left[H(z_a)\frac{z_b-z}{z_b-z_a}+H(z_b)\frac{z-z_a}{z_b-z_a} \right] ,
\end{align}
which implies that $V(\phi^{-1}(z))/F(\phi^{-1}(z))$ dominates the concave majorant of $H$.

Under the CIR model, the class of  interval-type strategies does not  include all single threshold-type strategies. 
In particular, the minimum value that $a$ can take is $0$. If  $2\mu\theta <  \sigma^2$, then $Y$ can reach level $0$ and reflects. The interval-type strategy with $a=0$ implies stopping the process $Y$ at level $0$, even though it could be optimal to wait and let $Y$ evolve.

Hence, we must also consider separately the candidate strategy of waiting for $Y$ to reach an  upper level $b\geq y$ without a lower stopping level. The well-known supermartingale property of $(e^{-rt}V(Y_t))_{t\ge 0}$ (see Appendix D of \cite{Karatzas1998a}) implies that $V(y)\geq \E_y\{e^{-r\tau}V(Y_\tau)\}$ for $\tau\in\setT$. Then,  taking $\tau = \tau_b$, we have
\begin{align*}
V(y)\geq \E_y\{e^{-r\tau_b}V(Y_{\tau_b})\} = V(b)\frac{F(y)}{F(b)},
\end{align*}
or equivalently,
\begin{align}\label{VbCIR}
\frac{V(\phi^{-1}(z))}{F(\phi^{-1}(z))} = \frac{V(y)}{F(y)} \geq \frac{V(b)}{F(b)} = \frac{V(\phi^{-1}(z_b))}{F(\phi^{-1}(z_b))},
\end{align}
which indicates that $V(\phi^{-1}(z))/F(\phi^{-1}(z))$ is \emph{decreasing}. By \eqref{VabCIR} and \eqref{VbCIR}, we now see that $V(y)\geq F(y)W(\phi(y))$, where $W$ is the \emph{decreasing} smallest concave majorant of $H$.

For the reverse inequality, we first show that
\begin{align}\label{eq:CIRop}
F(y)W(\phi(y))\geq \E_y\{e^{-r(t\wedge \tau)}F(Y_{t\wedge \tau})W(\phi(Y_{t\wedge \tau}))\},
\end{align}
for $y\in [0,+\infty)$, $\tau\in\setT$ and $t\geq 0$. If the initial value $y=0$, then    the decreasing property of $W$ implies the inequality
\begin{align}\label{eq:CIRop0}
\E_0\{e^{-r({t\wedge \tau})}F(Y_{t\wedge \tau})W(\phi(Y_{t\wedge \tau}))\}\leq \E_0\{e^{-r(t\wedge \tau)}F(Y_{t\wedge \tau})\}W(\phi(0))=F(0)W(\phi(0)),
\end{align}
where  the equality follows from the martingale property of $(e^{-rt}F(Y_t))_{t\ge 0}$.


When $y>0$,  the  concavity of $W$ implies that, for any   fixed $z$, there exists an affine function $L_z(\alpha):=m_z \alpha +c_z$ such that $L_z(\alpha) \geq W(\alpha)$ for $\alpha\geq\phi(0)$ and $L_z(z)=W(z)$ at $\alpha=z$, with constants $m_z$ and $c_z$. In turn, this yields  the inequality
\begin{align}
\E_y&\{e^{-r(\tau_0\wedge t\wedge\tau)}F(Y_{\tau_0\wedge t\wedge\tau})W(\phi(Y_{\tau_0\wedge t\wedge\tau}))\} \label{eq:CIRop4}\\
&\leq \E_y\{e^{-r(\tau_0\wedge t\wedge\tau)}F(Y_{\tau_0\wedge t\wedge\tau})L_{\phi(y)}(\phi(Y_{\tau_0\wedge t\wedge\tau}))\} \notag\\
&= m_{\phi(y)}\E_y\{e^{-r(\tau_0\wedge t\wedge\tau)}F(Y_{\tau_0\wedge t\wedge\tau})\phi(Y_{\tau_0\wedge t\wedge\tau})\} + c_{\phi(y)} \E_y\{e^{-r(\tau_0\wedge t\wedge\tau)}F(Y_{\tau_0\wedge t\wedge\tau})\}\notag\\
&= -m_{\phi(y)}\E_y\{e^{-r(\tau_0\wedge t\wedge\tau)}G(Y_{\tau_0\wedge t\wedge\tau})\}+ c_{\phi(y)} \E_y\{e^{-r(\tau_0\wedge t\wedge\tau)}F(Y_{\tau_0\wedge t\wedge\tau})\}\notag\\
&= -m_{\phi(y)}G(y) + c_{\phi(y)} F(y)\label{2}\\
&= F(y) L_{\phi(y)}(\phi(y))\notag\\
&= F(y)W(\phi(y)),\label{eq:CIRop1}
\end{align}
where $\eqref{2}$ follows from the martingale property of $(e^{-rt}F(Y_t))_{t\ge0}$ and $(e^{-rt}G(Y_t))_{t\ge0}$. If $2\mu\theta\geq \sigma^2$, then $\tau_0=+\infty$ for $y>0$. This immediately reduces \eqref{eq:CIRop4}-\eqref{eq:CIRop1}  to the desired inequality \eqref{eq:CIRop}.

On the other hand, if $2\mu\theta< \sigma^2$, then we decompose   \eqref{eq:CIRop4} into two terms:
\begin{align}
&\E_y\{e^{-r(\tau_0 \wedge t\wedge\tau)}F(Y_{\tau_0 \wedge t\wedge\tau})W(\phi(Y_{\tau_0 \wedge t\wedge\tau}))\} \notag\\
&=\underbrace{\E_y\{e^{-r(t\wedge \tau)}F(Y_{t\wedge \tau})W(\phi(Y_{t\wedge \tau}))\indic{t\wedge \tau\leq \tau_0}\}}_{\text{(I)}} + \underbrace{\E_y\{e^{-r\tau_0}F(Y_{\tau_0})W(\phi(Y_{\tau_0}))\indic{t\wedge \tau>\tau_0}\}}_{\text{(II)}}. \label{eq:CIRop2}
\end{align}
By the optional sampling theorem and   decreasing property of $W$,  the second term satisfies
\begin{align}
\text{(II)}  &= W(\phi(0))\E_y\{e^{-r\tau_0}F(Y_{\tau_0})\indic{t\wedge \tau>\tau_0}\} \notag\\
&\geq W(\phi(0))\E_y\{e^{-r(t\wedge\tau)}F(Y_{t\wedge\tau})\indic{t\wedge \tau>\tau_0}\}\notag\\
& \geq \E_y\{e^{-r(t\wedge\tau)}F(Y_{t\wedge\tau})W(\phi(Y_{t\wedge\tau}))\indic{t\wedge \tau>\tau_0}\}=: \text{(II')}.\label{eq:CIRop3}
\end{align}
Combining \eqref{eq:CIRop3} with   \eqref{eq:CIRop1}, we arrive at
\begin{align*}
F(y)W(\phi(y))  \ge \text{(I)} + \text{(II')} = \E_y\{e^{-r(t\wedge \tau)}F(Y_{t\wedge \tau})W(\phi(Y_{t\wedge \tau}))\},
\end{align*}
for all $y>0$. In all,    inequality \eqref{eq:CIRop} holds for all $y \in [0,+\infty)$, $\tau \in \setT$ and  $t\geq 0$. From \eqref{eq:CIRop} and the fact that $W$ majorizes $H$, it follows that
\begin{align}
F(y)W(\phi(y)) &\geq \E_y\{e^{-r(t\wedge \tau)}F(Y_{t\wedge \tau})W(\phi(Y_{t\wedge \tau}))\} \\
&\geq \E_y\{e^{-r(t\wedge \tau)}F(Y_{t\wedge \tau})H(\phi(Y_{t\wedge \tau}))\}  \geq \E_y\{e^{-r(t\wedge \tau)}h_s(Y_{t\wedge \tau})\}.\label{fdsa}
\end{align}
Maximizing \eqref{fdsa} over all $\tau\in\setT$ and $t\geq 0$ yields the reverse inequality $F(y)W(\phi(y)) \geq V(y)$.
\end{proof}

In summary, we have found  an expression for the value function  $V(y)$ in  \eqref{Vchi}, and proved that it is sufficient to consider only candidate stopping times   described by the first   time   $Y$ reaches  a single upper threshold or exits an interval. To determine the optimal timing strategy, we  need to understand the properties of $H$ and its smallest concave majorant $W$. To this end, we have the following lemma.

\begin{lemma}\label{lm:HCIR}
The function $H$ is continuous on $[\phi(0), 0]$, twice differentiable on $(\phi(0), 0)$ and possesses the following properties:
\begin{enumerate}[(i)]
\item \label{HCIR0} $H(0)=0$, and
\begin{align}\label{eq:HCIR0}
H(z)\begin{cases}
<0 &\, \textrm{ if }\, z \in [\phi(0), \phi(c_s)),\\
>0 &\, \textrm{ if }\, z \in (\phi(c_s),0).
\end{cases}
\end{align}
\item \label{HCIR1}
$H(z)$ is strictly increasing for $z \in (\phi(0),\phi(c_s)\vee\phi(y_s))$.
\item \label{HCIR2}
\begin{align*}
H(z) \textrm{ is }
\begin{cases}
\textrm{convex} &\, \textrm{ if }\, z \in (\phi(0), \phi(y_s)],\\
\textrm{concave} &\, \textrm{ if }\, z \in [\phi(y_s),0).
\end{cases}
\end{align*}
\end{enumerate}
\end{lemma}

 In Figure \ref{fig:HCIR}, we see that  $H$ is first increasing then decreasing, and first convex  then concave. Using these properties, we now derive the optimal stopping timing.

\begin{figure}[h]
\begin{center}
 \subfigure[\small{$2\mu\theta < \sigma^2$}]
 {\label{fig:HCIR_rf}
 \scalebox{0.39} {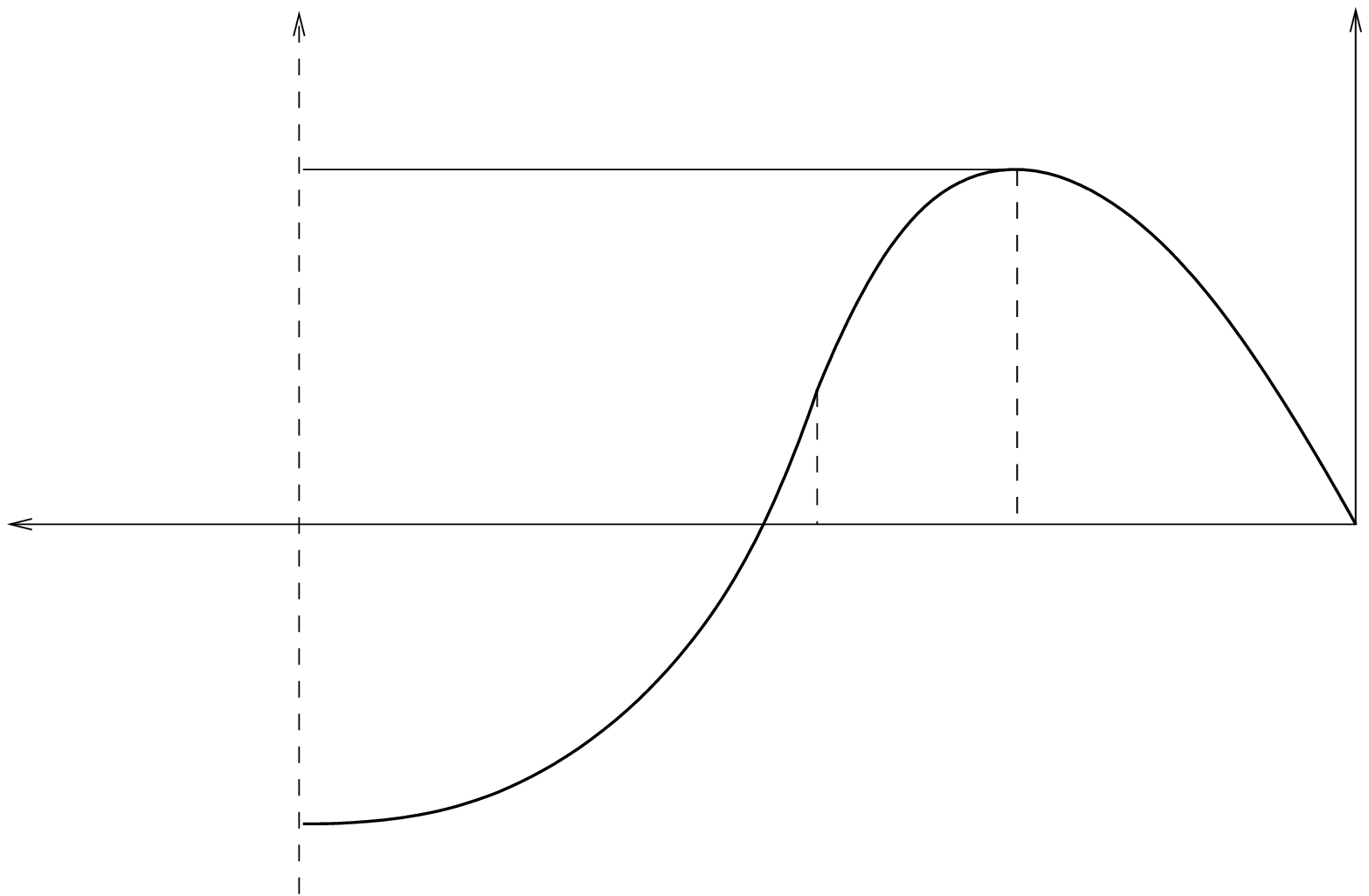}}
 ~~~~\qquad\subfigure[\small{$2\mu\theta \geq\sigma^2$ }]
 {\label{fig:HCIR_en}\scalebox{0.35}{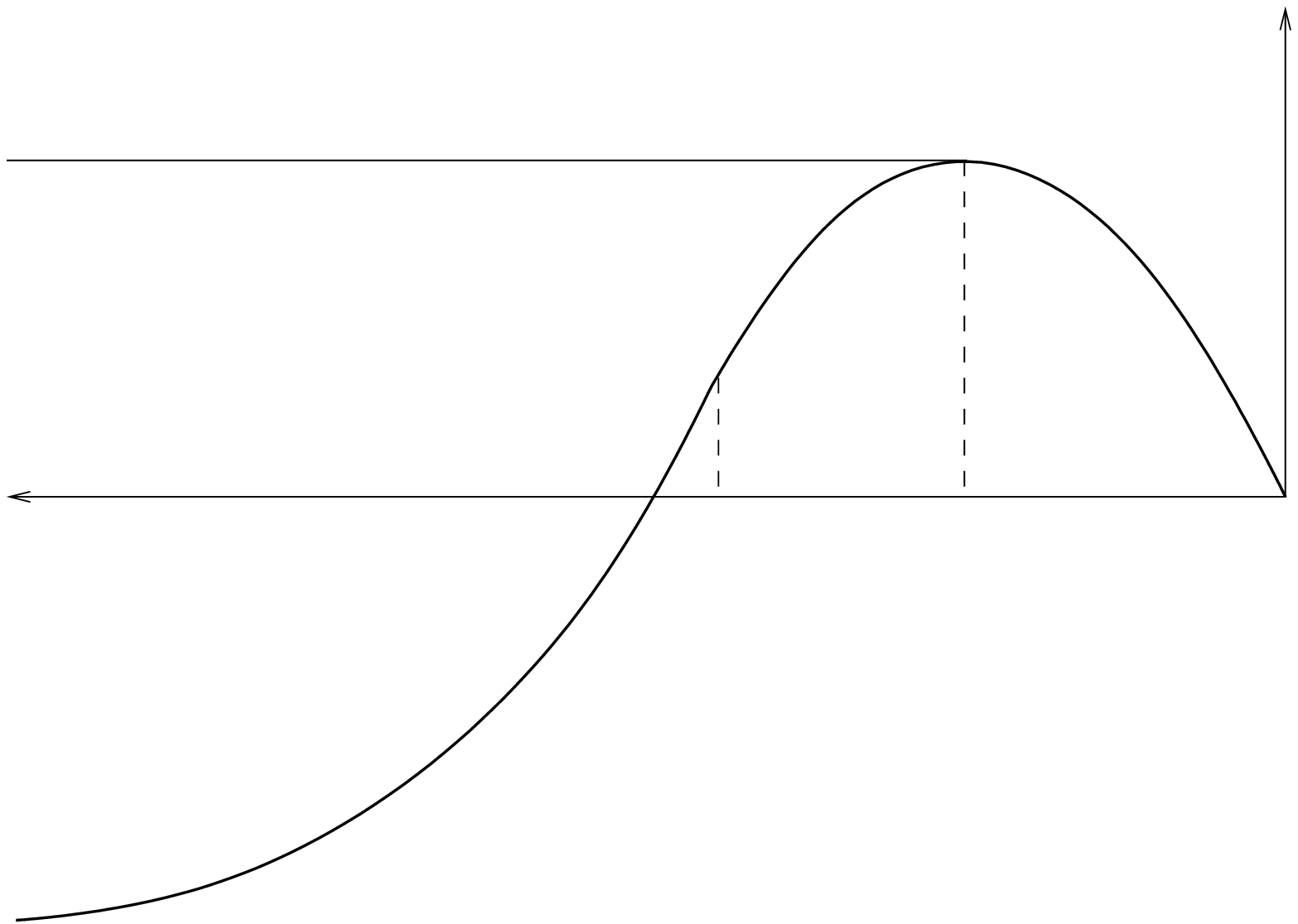}}
\end{center}
\setlength{\abovecaptionskip}{-5pt}
\caption{\small{Sketches of $H$ and $W$. The function $W$ is equal to the constant $H(z^*)$ on $(\phi(0),z^*)$, and coincides with $H$ on $[z^*,0]$. Note that $-\infty \!<\! \phi(0) \!<\!0$ if $2\mu\theta < \sigma^2$, and $\phi(0)\! =\! -\infty$ if $2\mu\theta \geq\sigma^2$.}}
\label{fig:HCIR}
\end{figure}

\paragraph*{Proof of Theorem \ref{thm:optLiquCIR}}
We  determine the value function in the form: $V(y)  =
F(y)W(\phi(y))$, where $W$ is the decreasing
smallest concave majorant of $H$. By Lemma \ref{lm:HCIR} and
Figure \ref{fig:HCIR}, $H$ peaks at  $z^* > \phi(c_s)
\vee \phi(y_s)$ so that
\begin{align}\label{eq:HzLiquCIR}
{H}^\prime\!(z^*)=0.
\end{align}
In turn, the decreasing  smallest concave majorant admits the form:
\begin{align}\label{WCIR}
W(z) = \begin{cases}
H(z^*) &\, \textrm{ if }\, z < z^*,\\
H(z) &\, \textrm{ if }\, z \geq z^*.
\end{cases}
\end{align}
Substituting $b^* = \phi^{-1}(z^*)$ into \eqref{eq:HzLiquCIR}, we have
\begin{align*}
{H}^\prime\!(z^*) &= \frac{F(\phi^{-1}(z^*))-(\phi^{-1}(z^*)-c_s){F}^\prime\!(\phi^{-1}(z^*))}{{F}^\prime\!(\phi^{-1}(z^*))G(\phi^{-1}(z^*)) - F(\phi^{-1}(z^*)) {G}^\prime\!(\phi^{-1}(z^*))}\\
&= \frac{F(b^*) - (b^*-c_s){F}^\prime\!(b^*)}{{F}^\prime\!(b^*)G(b^*)-F(b^*){G}^\prime\!(b^*)},
\end{align*}
which can be further simplified to \eqref{eq:solvebCIR}. We can express $H(z^*)$ in terms of $b^*$:
\begin{align}\label{HzCIR}
H(z^*) = \frac{b^*-c_s}{F(b^*)}.
\end{align}
Applying \eqref{HzCIR} to \eqref{WCIR}, we get
\begin{align*}
W(\phi(y)) = \begin{cases}
H(z^*) = \frac{b^*-c_s}{F(b^*)} &\, \textrm{ if }\, y < b^*,\\
H(\phi(y)) = \frac{y-c_s}{F(y)} &\, \textrm{ if }\, y \geq b^*.
\end{cases}
\end{align*}
Finally, substituting this into the value function $V(y)  = F(y)W(\phi(y))$, we conclude.

\subsubsection{Optimal Starting Timing}\label{sect-CIRentryproof}
We now turn to the optimal starting problem. Our methodology in Section \ref{sect-CIRexitproof} applies to general payoff functions, and thus can be applied to the optimal starting problem \eqref{JCIR} as well. To this end, we apply the same transformation \eqref{phi} and define the function
\begin{align*}
\hat{H}(z)  := \begin{cases}
\frac{\hat{h}}{F}\circ \phi^{-1}(z) &\, \textrm{ if }\, z<0,\\
\lim_{y\to +\infty}\limits\frac{(\hat{h}(y))^+}{F(y)} &\, \textrm{ if }\, z=0,
\end{cases}
\end{align*}
where $\hat{h}$ is given in \eqref{hhat}. We then follow  Theorem \ref{thm:optLiquCIR} to determine the value function $J$. This amounts to  finding  the \emph{decreasing} smallest concave majorant $\hat{W}$ of $\hat{H}$.  Indeed, we can   replace $H$ and $W$  with $\hat{H}$ and $\hat{W}$ in Theorem \ref{thm:optLiquCIR} and its proof. As a result, the value function of the optimal starting timing problem must take the form \[J(y)=F(y)\hat{W}(\phi(y)).\]

To solve the optimal starting timing problem, we need to understand the properties of $\hat{H}$.
\begin{lemma}\label{lm:hatHCIR}
The function $\hat{H}$ is continuous on $[\phi(0),0]$, differentiable on $(\phi(0),0)$, and twice differentiable on $(\phi(0),\phi(b^*)) \cup (\phi(b^*),0)$, and possesses the following properties:
\begin{enumerate}[(i)]
\item \label{hatHCIR0}
$\hat{H}(0)=0$. Let $\bar{d}$ denote the unique solution to $\hat{h}(y)=0$, then  $\bar{d}<b^*$ and
\begin{align*}
\hat{H}(z) \begin{cases}
>0 &\, \textrm{ if }\, z \in [\phi(0), \phi(\bar{d})),\\
< 0 &\, \textrm{ if }\, z \in (\phi(\bar{d}), 0).
\end{cases}
\end{align*}

\item \label{hatHCIR1}
$\hat{H}(z)$ is strictly increasing for $z > \phi(b^*)$ and $\lim_{z\to \phi(0)}\hat{H}^{'}\!(z) = 0$.

\item \label{hatHCIR2}
\begin{align*}
\hat{H}(z) \textrm{ is }
\begin{cases}
\textrm{concave} &\, \textrm{ if }\, z \in (\phi(0), \phi(y_b)),\\
\textrm{convex} &\, \textrm{ if }\, z \in (\phi(y_b),0).
\end{cases}
\end{align*}
\end{enumerate}
\end{lemma}
By Lemma \ref{lm:hatHCIR}, we sketch $\hat{H}$ in Figure \ref{fig:hatHCIR}.

\begin{figure}[h]
\begin{center}
 \subfigure[\small{$2\mu\theta < \sigma^2$ }]
 {\label{fig:hatHCIR_rf}
 \scalebox{0.3} {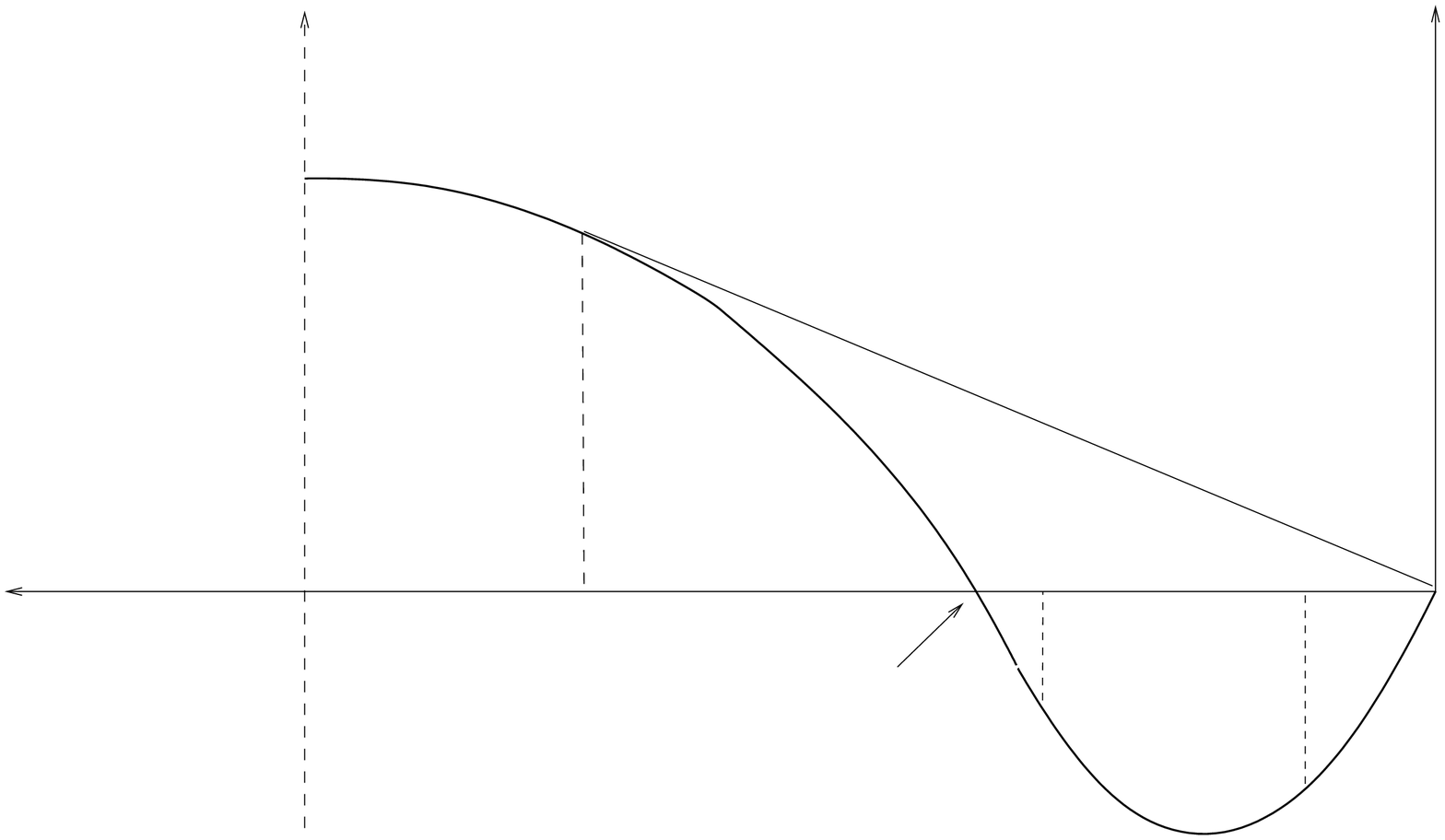}}
 \subfigure[\small{$2\mu\theta \geq\sigma^2$ }]
 {\label{fig:hatHCIR_en}\scalebox{0.3}{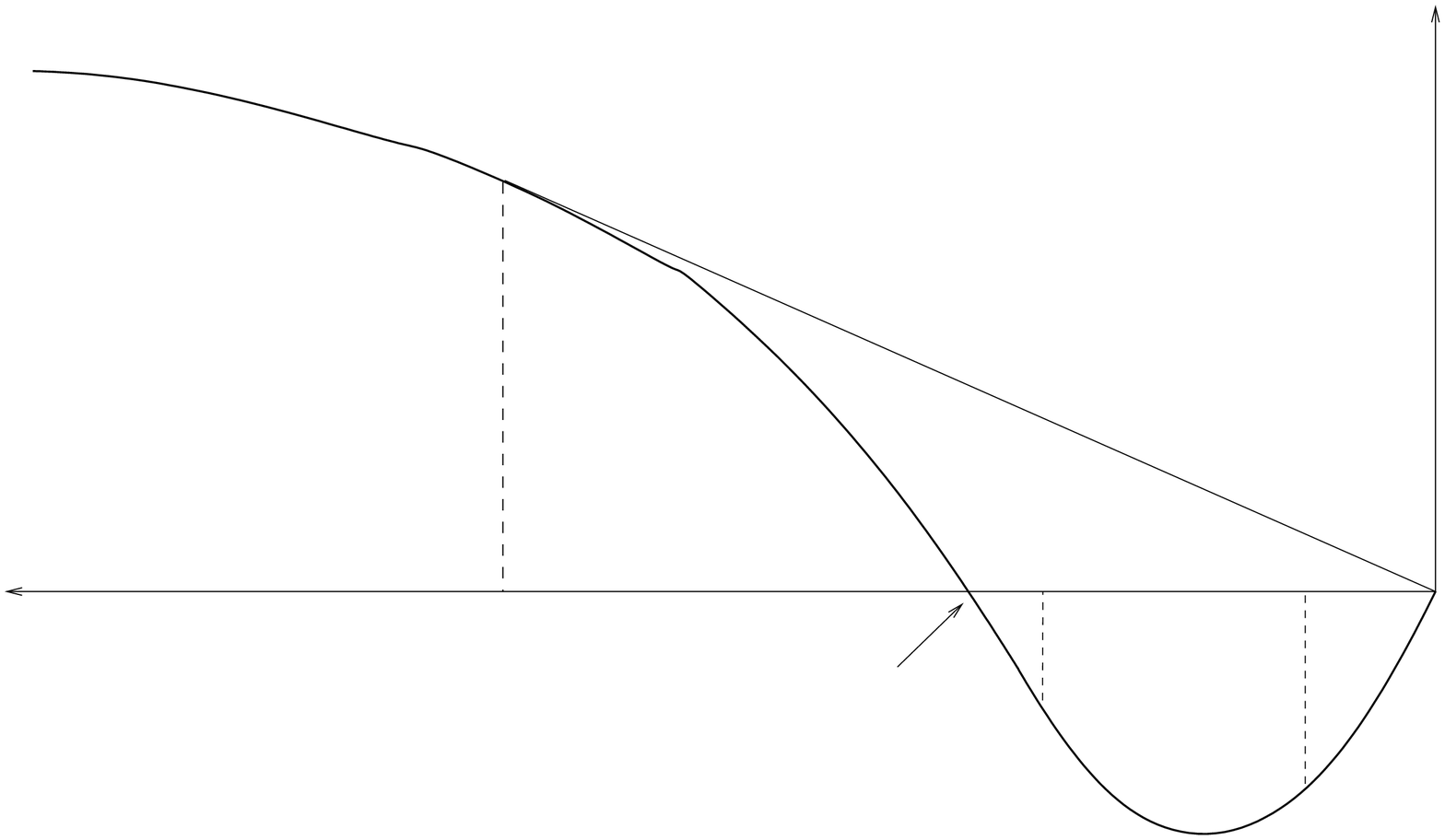}}
\end{center}
\setlength{\abovecaptionskip}{-5pt}
\caption{\small{Sketches of $\hat{H}$ and $\hat{W}$. The function $\hat{W}$ coincides with $\hat{H}$ on $[\phi(0), \hat{z}]$ and is a straight line tangent to $\hat{H}$ at $\hat{z}$ on $(\hat{z},0]$. Note that $-\infty\!<\!\phi(0)\!<\!0$ if $2\mu\theta < \sigma^2$, and $\phi(0)\!=\!-\infty$ if $2\mu\theta \geq\sigma^2$.}}
\label{fig:hatHCIR}
\end{figure}


\paragraph*{Proof of Theorem \ref{thm:optEntryCIR}}
To determine  the value function in the form: $J(y)  = F(y)\hat{W}(\phi(y))$, we analyze  the decreasing smallest concave majorant,  $\hat{W}$, of $\hat{H}$. By Lemma \ref{lm:hatHCIR} and Figure \ref{fig:HCIR},  we have $\hat{H}^{'}\!(z) \to 0$ as $z \to \phi(0)$. Therefore, there exists a unique number $\hat{z} \in (\phi(0), \phi(b^*))$ such that
\begin{align}\label{eq:HzEntryCIR}
\frac{\hat{H}(\hat{z})}{\hat{z}} = \hat{H}^{'}\!(\hat{z}).
\end{align}
In turn, the decreasing smallest concave majorant admits the form:
\begin{align}\label{hatWCIR}
\hat{W}(z) = \begin{cases}
\hat{H}(z) &\, \textrm{ if }\, z \leq \hat{z},\\
z\frac{\hat{H}(\hat{z})}{\hat{z}} &\, \textrm{ if }\, z > \hat{z}.
\end{cases}
\end{align}
Substituting $d^* = \phi^{-1}(\hat{z})$ into \eqref{eq:HzEntryCIR}, we have
\begin{align}\label{hatHzzCIR}
\frac{\hat{H}(\hat{z})}{\hat{z}} = \frac{\hat{H}(\phi(d^*)}{\phi(d^*)} = -\frac{V(d^*)-d^*-c_b}{G(d^*)},
\end{align}
and
\begin{align*}
\hat{H}^{'}\!(\hat{z}) = \frac{F(d^*)({V}^\prime\!(d^*)  - 1) - F^{'}\!(d^*)(V(d^*)-(d^*+c_b))}{F^{'}\!(d^*)G(d^*)-F(d^*)G^{'}\!(d^*)}.
\end{align*}
Equivalently, we can express condition \eqref{eq:HzEntryCIR} in terms of $d^*$:
\begin{align*}
-\frac{V(d^*)-(d^*+c_b)}{G(d^*)} = \frac{F(d^*)({V}^\prime\!(d^*)  - 1) - F^{'}\!(d^*)(V(d^*)-(d^*+c_b))}{F^{'}\!(d^*)G(d^*)-F(d^*)G^{'}\!(d^*)},
\end{align*}
which shows $d^*$ satisfies \eqref{eq:solvedCIR} after simplification.

Applying \eqref{hatHzzCIR} to \eqref{hatWCIR}, we get
\begin{align*}
W(\phi(y)) = \begin{cases}
\hat{H}(\phi(y)) = \frac{V(y)-(y+c_b)}{F(y)} &\, \textrm{ if }\, y \in [0, d^*],\\
\phi(y) \frac{\hat{H}(\hat{z})}{\hat{z}} = \frac{V(d^*)-(d^*+c_b)}{G(d^*)}\frac{G(y)}{F(y)} &\, \textrm{ if }\, y \in  (d^*,+\infty).
\end{cases}
\end{align*}
From this, we obtain the value function.

\subsection{Optimal Switching Problem}

\paragraph*{Proofs of Theorems \ref{thm:CIRth2} and \ref{thm:CIRth3}}
\cite{zervos2011buy} have studied a similar problem of trading a mean-reverting asset with fixed transaction costs, and provided detailed proofs using a variational inequalities approach. In particular, we 
observe that $y_b$ and $y_s$ in \eqref{CIRLrbs}  play the  same roles as $x_b$ and $x_s$ in Assumption 4 in \cite{zervos2011buy}, respectively. However, Assumption 4 in \cite{zervos2011buy} requires  that $0 \leq x_b,$ this is not necessarily true for $y_b$ in our problem. We have checked and realized that this assumption is not necessary for Theorem \ref{thm:CIRth2}, and that $y_b < 0$ simply implies that there is no optimal starting level, i.e. it is never optimal to start.

In addition, \cite{zervos2011buy} assume (in their Assumption 1) that the hitting time of level $0$ is infinite with probability 1. In comparison, we consider not only the CIR case where 0 is inaccessible, but also when the CIR process has a reflecting boundary at $0$. In fact, we find that the proofs in \cite{zervos2011buy} apply to both cases under the CIR model.  Therefore, apart from relaxation of the aforementioned assumptions, the proofs of our Theorems \ref{thm:CIRth2} and \ref{thm:CIRth3} are the same as that of Lemmas 1 and 2 in \cite{zervos2011buy} respectively.

\appendix
\section{Appendix}
\noindent \textbf{A.4 ~Proof of Lemma \ref{lm:HCIR} (Properties of $H$).}\, \label{pf-CIR-H}
\noindent (i) First, we compute
\begin{align*}
H(0) = \lim_{y\to+\infty}\frac{(h_s(y))^+}{F(y)} =  \lim_{y\to+\infty}\frac{y-c_s}{F(x)} =\lim_{y\to+\infty}\frac{1}{{F}^\prime\!(y)} =0.
\end{align*}
Using the facts that  $F(y) >0$ and  $\phi(y)$ is a strictly increasing function, \eqref{eq:HCIR0} follows.

\noindent (ii) We look at the first derivative of $H$:
\begin{align*}
{H}^\prime\!(z) = \frac{1}{\phi'(y)} (\frac{h_s}{F})'(y) =  \frac{1}{\phi'(y)}\frac{F(y) - (y-c_s){F}^\prime\!(y)}{{F}^2(y)}, \quad z=\phi(y).
\end{align*}
Since both $\phi'(y)$ and ${F}^2(y)$ are positive,  it remains to determine the sign of $F(y) - (y-c_s){F}^\prime\!(y)$. Since ${F}^\prime\!(y) >0$, we can equivalently check  the sign of $v(y) :=\frac{F(y)}{{F}^\prime\!(y)} -(y-c_s)$. Note that $v'(y) = -\frac{F(y){F}^{\prime\prime}\!(y)}{({F}^\prime\!(y))^2} <0$. Therefore, $v(y)$ is a strictly decreasing function. Also, it is clear that $v(c_s) >0$ and $v(y_s)>0$. Consequently,  $v(y)>0$ if $y< (c_s \vee y_s)$ and hence, $H(z)$ is strictly increasing if $z \in (\phi(0),\phi(c_s)\vee\phi(y_s))$.


\noindent (iii) By differentiation, we have
\begin{align*}
H^{''}\!(z) = \frac{2}{\sigma^2F(y)(\phi'(y))^2}(\L-r)h_s(y),\quad z=\phi(y).
\end{align*}
Since $\sigma^2, F(y)$ and $(\phi'(y))^2$ are all positive, the convexity/concavity  of  $H^{''}$ depends on  the sign of
\begin{align*}
(\L - r)h_s(y) = \mu(\theta - y) - r(y-c_s) = (\mu\theta+rc_s)-(\mu+r)y
&\begin{cases}
\geq 0 &\, \textrm{ if }\, y \in [0, y_s],\\
\leq 0 &\, \textrm{ if }\, y \in [y_s, +\infty),
\end{cases}
\end{align*}
which implies property (iii).
$\scriptstyle{\blacksquare}$\\

\noindent \textbf{A.5 ~Proof of Lemma \ref{lm:hatHCIR} (Properties of $\hat{H}$).}\, \label{pf-CIR-hatH}
It is straightforward to check that $V(y)$ is continuous and differentiable everywhere, and  twice differentiable everywhere except at $y=b^*$, and all these holds for $\hat{h}(y) = V(y) - (y + c_b)$. Both $F$ and $\phi$ are twice differentiable.  In turn, the continuity and differentiability of $\hat{H}$ on $(\phi(0),0)$ and twice differentiability of $\hat{H}$ on $(\phi(0),\phi(b^*)) \cup (\phi(b^*),0)$ follow.

To show the continuity of $\hat{H}$ at $0$, we note that \begin{align*}\hat{H}(0)&=\lim_{y\to+\infty}\frac{(\hat{h}(y))^+}{F(y)}= \lim_{y\to+\infty}\frac{0}{F(y)}=0, \quad \text{ and }\\
\lim_{z \rightarrow 0}\hat{H}(z) &= \lim_{y\to+\infty}\frac{\hat{h}}{F}(y)= \lim_{y\to+\infty}\frac{-(c_s+c_b)}{F(y)} =0.
\end{align*}
From this, we conclude that  $\hat{H}$ is also continuous at $0$.

\noindent (i) First, for $y\in [b^*,+\infty)$,  $\hat{h}(y)\equiv -(c_s+c_b)$. For $y\in (0, b^*)$, we compute
\begin{align*}
{V}^\prime\!(y) = \frac{b^*-c_s}{F(b^*)} {F}^\prime\!(y) = \frac{{F}^\prime\!(y)}{{F}^\prime\!(b^*)},\quad \textrm{by \eqref{eq:solvebCIR}}.
\end{align*}
Recall that ${F}^\prime\!(y)$ is a strictly increasing function  and  $\hat{h}(y) = V(y)-y-c_b$.  Differentiation yields
\begin{align*}
\hat{h}^{'}\!(y) = {V}^\prime\!(y)-1 = \frac{{F}^\prime\!(y)}{{F}^\prime\!(b^*)} - 1 < \frac{{F}^\prime\!(b^*)}{{F}^\prime\!(b^*)}-1=0, \quad y \in (0, b^*),
\end{align*}
which implies that $\hat{h}(y)$ is strictly decreasing for $y\in (0,b^*)$. On the other hand,
 $\hat{h}(0) >0$ as we are considering  the non-trivial case. Therefore, there exists a unique solution $\bar{d}<b^*$ to $\hat{h}(y) = 0$, such that $\hat{h}(y)>0$ for $y \in [0, \bar{d})$, and $\hat{h}(y)<0$ for $y \in (\bar{d},+\infty)$.
With $\hat{H}(z) = (\hat{h}/F) \circ \phi^{-1}(z)$, the above properties of $\hat{h}$, along with the facts  that  $\phi(y)$ is strictly increasing and   $F(y) > 0$, imply  property (i).

\noindent (ii)With $z=\phi(y)$, for $y>b^*$, $\hat{H}(z)$ is strictly increasing since
\begin{align*}
\hat{H}^{'}\!(z) = \frac{1}{\phi'(y)} (\frac{\hat{h}}{F})'(y) =  \frac{1}{\phi'(y)}(\frac{-(c_s+c_b)}{F(y)})'=\frac{1}{\phi'(y)}\frac{(c_s+c_b)F^{'}\!(y)}{F^{\chi^2}(y)} > 0.
\end{align*}
When $y\to0$,  because $(\frac{\hat{h}(y)}{F(y)})'$ is finite, but $\phi'(y)\to+\infty$, we have $\lim_{z\to\phi(0)} \hat{H}^{'}\!(z) = 0$.

\noindent (iii) Consider the second derivative:
\begin{align*}\label{eq:hatHCIR2}
\hat{H}^{''}\!(z) = \frac{2}{\sigma^2F(y)(\phi'(y))^2}(\L-r)\hat{h}(y).
\end{align*}
The positivities of  $\sigma^2, F(y)$ and $(\phi'(y))^2$ suggest that we inspect the sign of $(\L - r)\hat{h}(y)$:
\begin{align*}
(\L - r)\hat{h}(y) &= \half \sigma^2x {V}^{\prime\prime}\!(y) + \mu(\theta -y) {V}^\prime\!(y) - \mu (\theta-y) - r(V(y) -(y +c_b))\\
& = \begin{cases}
(\mu+r)y-\mu\theta +rc_b &\, \textrm{ if }\, y < b^*,\\
r(c_s+c_b)>0 &\, \textrm{ if }\, y > b^*.
\end{cases}
\end{align*}
Since $\mu, r > 0$ by assumption, $(\L - r)\hat{h}(y)$ is strictly increasing  on $(0,b^*)$. Next, we show that $0\!<\!y_b \!<\! y_s \!<\! b^*$. By the fact that ${F}^\prime\!(0) = \frac{r}{\mu\theta}$ and the assumption that $V(0) = \frac{b^*-c_s}{F(b^*)} > c_b$, we have
\begin{align*}
{V}^\prime\!(0) =\frac{b^*-c_s}{F(b^*)} {F}^\prime\!(0)  =\frac{b^*-c_s}{F(b^*)}\frac{r}{\mu\theta} >  \frac{rc_b}{\mu\theta}.
\end{align*}
In addition, by the convexity of $V$ and ${V}^\prime\!(b^*)=1$, it follows that
\begin{align*}
\frac{rc_b}{\mu\theta} < {V}^\prime\!(0) < {V}^\prime\!(b^*)=1,
\end{align*}
which implies $\mu\theta > rc_b$ and hence $y_b>0$. By simply comparing the definitions of  $y_b$ and $y_s$, it is clear that $y_b < y_s$. Therefore, by observing that $(\L - r)\hat{h}(y_b)=0$, we conclude $(\L - r)\hat{h}(y) <0$ if $y \in [0, y_b)$, and $(\L - r)\hat{h}(y)>0$ if $y \in (y_b, +\infty)$. This suggests the concavity and convexity of $\hat{H}$ as desired.
$\scriptstyle{\blacksquare}$

\bibliographystyle{apa}
 \singlespacing
 
\bibliography{mybib_2012}



\end{document}